\documentclass[11pt]{article}

\usepackage{amsthm,amsmath,amssymb,bm}
\usepackage{graphicx,color,epsfig}
\usepackage[margin=1in]{geometry}

\newcommand{\N}{\mathbb{N}}
\newcommand{\R}{\mathbb{R}}

\newcommand{\Z}{\mathbb{Z}}

\newtheorem{Thm}{Theorem}[section]
\newtheorem{Claim}[Thm]{Claim}
\newtheorem{Pro}[Thm]{Proposition}

\begin{document}

\title{Dynamics of buyer populations in fresh product markets}
\author{Ali Ellouze and Bastien Fernandez}
\date{}
\maketitle

\begin{center}
Laboratoire de Probabilit\'es, Statistique et Mod\'elisation\\
CNRS - Univ. Paris Cit\'e -  Sorbonne Univ.\\
Paris, France\\
ellouze@lpsm.paris and fernandez@lpsm.paris
\end{center}

\begin{abstract}
Based on empirical evidences and previous studies, we introduce and mathematically study a perception-driven model for the dynamics of  buyer populations in markets of perishable goods. Buyer behaviours are driven partly by some loyalty to the sellers that they previously purchased at, and partly by the sensitivity to the intrinsic attractiveness of each seller in the market. On the other hand, the sellers update they attractiveness in time according to the difference between the volume of their clientele and the mean volume of buyers in the market, optimising either their profit when this difference is favourable or their competitiveness otherwise. While this negative feedback mechanism is a source of instability that promotes oscillatory behaviour, our analysis identifies the critical features of the dynamics that are responsible for the asymptotic stability of the stationary states, both in their immediate neighbourhood and globally in phase space. 
Altogether, this study provides mathematical insights into the consequences of introducing feedback into buyer-seller interactions in such markets, with emphasis on identifying conditions for long term constancy of clientele volumes.
\end{abstract}

\leftline{\small\today.}

\section{Introduction and definition of the model}
\subsection{Background and context}
In addition to their historical and cultural importance in rural societies \cite{G78}, fresh product and foodstuff markets also play a central role in the supply chain of contemporary urban populations, including at the megalopolis scale. Some of the largest examples of such wholesale markets are the Rungis International Market near Paris, the Central de Abastos in Mexico City and Mercamadrid in Madrid, which each secure daily food supply for tens of millions of people \cite{MarketsWebSites}.  

Many such markets function based on pairwise buyer-seller interactions, without posted prices and/or include bargaining as a standard practice.
In that case, buyer behaviours are known to be driven by some degree of loyalty to the previously visited merchants and, concurrently, by probing for best profit \cite{C10,CTG12,K01,R13,VE11,WKH00}. Loyalty simply means the trend to return to the same sellers that were traded with at the previous market instance (e.g.\ previous week-day or previous week). In absence of public or solid prices, best profit is a more convoluted notion that combines seller (stock-dependent) inclination towards negotiation  with product quality and demand, and also sometimes, with accompanying services. In addition, opportunities may also vary in time depending on the merchant reaction to clientele changes \cite{G05}.

As mathematical modelling is concerned, various attempts to capture the behavioural dynamics in fresh product wholesale markets have been developed in the literature. In particular, a fully analysable simple model has been proposed for the time evolution of the preference towards certain merchants over others \cite{WKH00}. This dynamics includes dependence on the potential profit to buy from a given merchant, and shows interesting parameter-dependent bifurcations in the phase portrait of the asymptotic functioning modes. However, the profits are assumed to be constant in time and the model does not include variable feedback from the merchants. 

Besides, various agent-based detailed models have been introduced that incorporate several elements of the buyer-seller interactions, including feedback and bargaining. The accompanying numerical simulations have shown proved capacity to adequately reproduce the salient characteristics of the participants behaviour, such as persistent price dispersion and high loyalty \cite{HK95,KV00}. The numerics also have provided quantitative evaluations of various bargaining strategies \cite{CCB09,MR08}. Yet, the proposed high level of detail in these models prevents one from a mathematically rigorous and complete analysis.

\subsection{Essential characteristics of fresh product wholesale markets}
As a complement to previous studies, this paper aims to introduce and to study a perception-driven model of the dynamics of the buyer populations in fresh product wholesale markets, in presence of an adaptive feedback from the merchants. To some extent, the model can be regarded as an extension of the one in \cite{WKH00} - where buyers are partly faithful and partly sensitive to merchant's attractiveness - which includes dependence of this attractiveness on the (relative) volumes of clientele. 

Of particularly importance for our purpose, the model should be sufficiently simple to be amenable to a full mathematical description, and in particular, to identify the mechanisms that are responsible for the asymptotic-in-time phenomenology. Accordingly, the modelling process can only retain a limited number of features of these markets. Aiming at incorporating feedback in the dynamics, the following characteristics have been selected among the numerous {\em in situ} observations and interviews of the actors involved at Rungis market \cite{C10}.
\begin{itemize}
\item Every day, buyers arrive at the (say fruits and vegetables\footnote{For the sake of concurrence, we consider a single category of goods, which has to be broad enough in order to allow for a buyer to purchase goods of different types at several sellers the same day. For instance, the fruits and vegetables' category is a typical example.}) market. Starting with their regular seller(s) but also probing for best opportunity, the buyers inquire various sellers in order to evaluate multiple factors associated with potential profit, such as price, quality, willingness towards negotiation, accompanying services, etc. 
\item Even though a buyer purchases at a given (possibly their preferred) seller, he/she may also buy an additional part of their supply - some good(s) of different type in the same category in this case - at another seller, when a better opportunity occurs. Such purchases at various sellers on the same day is a typical intermediate pattern of switching behaviour from a preferred merchant at previous days to another one in the subsequent days. 
\item Transactions prices are not disclosed and remain private. Such confidentiality\footnote{However, market authorities survey prices based on  declarations by a representative sample of the wholesalers. The collected prices on a day are made public the next day, as a general information to the buyers.} leaves ample room for negotiations. Negotiations are affected by various extrinsic factors such as demand, (current and forthcoming) stock sizes and the prices paid by the wholesalers to their providers\footnote{Wholesalers may decide by themselves the prices paid to their providers, after their stock has been sold. This rule, usually adopted in the case of a long term trustworthy relationship with a provider, gives more flexibility and stability to wholesalers, who can adapted their own profit accordingly.}. 
\item Behaviours are influenced by the assessment of activity/clientele volume in the market. On one hand, a busy and active market signals the buyers that good opportunities might fade away if they do not conclude their purchases swiftly. In particular, a merchant with large clientele is perceived as being more dominant, and thus less prone to negotiation. On the other hand, a wholesaler with little clientele, especially when compared to the overall clientele at their competitors, is likely to increase their attractiveness. 
\end{itemize}

\subsection{Definition of the model}\label{S-DEF}
The modelling process in this paper does not intend to provide a faithful description of the economic/commercial aspects related to the functioning of fresh product wholesale markets. It more modestly aims at investigating how basic day-to-day actions and reactions of the various buyers and wholesalers, as they result from basic perceptions, can impact the volumes of clientele in the long term. 

Given the markets' main features presented above, the two populations are modelled as follows. 
\begin{itemize}
\item Wholesalers are represented by an intrinsic number, called {\bf attractiveness}, that encapsulates the multiple factors mentioned above of merchants' attraction power. In absence of solid information about prices, and given  that buyers are themselves professionals (retailers, restaurant managers/employees, etc) who may adapt their own prices to their clients {\em a posteriori} in order to maintain profit, we postulate that such a comprehensive notion might be more relevant  than the mere price when it comes to convincing a buyer to purchase at a given merchant.
\item Buyers are characterised by the (relative) {\bf volumes} of their population at each merchant. No other feature is envisaged. In particular, changes in merchant's attractiveness will only be affected by clientele volumes and will not depend on any other element such as the amount of demand or the stock size. In other words, this simplifying assumption considers that clientele size is the main element of buyer population in the markets under consideration, and hence that buyer demand and stock sizes are, in first approximation, directly correlated to such volume. 
\end{itemize}

\paragraph{Formal definition of the dynamical system under consideration.} 
The model itself is a (discrete time) dynamical system whose variables are quantifiers of the clientele volumes and of the wholesaler attractivenesses. A formal definition can be defined as follows. The market is assumed to be composed by $N\in \Z^+=\{1,2,\cdots\}$ competing merchants of a single product category, and to open at repetitive and regular events, typically every working day. The events are labelled by the discrete variable $t\in\N=\{0,1,2,\cdots\}$. 

The volume of clientele of merchant $i\in\{1,\cdots ,N\}$ at day $t$ is given by the fraction $p_i^t\in [0,1]$ of the total buyer population present in the market that day. Importantly, buyers can purchase distinct goods of the same category at distinct merchants the same day; hence the total of clientele fractions $\sum_{i=1}^Np_i^t$ can be {\sl a priori} any number in $[0,N]$, and it usually differs from the total buyer population, {\sl ie.}\  we do not impose the exclusive condition $\sum_{i=1}^Np_i^t=1$. 

The attractiveness of merchant $i$ at day $t$ is given by the positive coefficient $a_i^t\in\R_\ast^+$. When $a_i^t>1$, merchant $i$ is competitive and attracts prospecting customers; the higher $a_i^t$, the stronger the attraction. On the opposite, $a_i^t<1$ means a non-competitive seller that repels their clientele; the lower $a_i^t$, the stronger the repulsion. For the sake of the continuous dependence of the dynamics on the attractiveness, we also assume that in the intermediate case $a_i^t=1$, the merchant is neutral, {\sl ie.}\ it does not attract nor repel its clientele. Such an intermediate neutral status of a merchant is a consequence of the mathematical abstraction that has otherwise little relevance in modelling.

In summary, the state of the market on day $t$ is given by the $2N$-dimensional vector 
\[
(\mathbf p^t,\mathbf a^t)\in M_N:=[0,1]^N\times (\R_\ast^+)^N,
\]
where $\mathbf p^t=(p^t_1,\cdots ,p^t_N)$ and $\mathbf a^t=(a^t_1,\cdots ,a^t_N)$ are the variables of the dynamical system under consideration. 
Once the variables have been defined, the rule that governs their time evolution can be specified. Here, this rule is given by the following iterations that are inspired by the market characteristics exposed above
\begin{equation}
\left\{\begin{array}{l}
p_i^{t+1}=f_{\alpha,a_i^{t+1}}(p_i^{t})\\
a_i^{t+1}=a_i^{t}g(p_i^{t},\frac1{N}\sum_{j=1}^Np_j^{t})
\end{array}\right.\ \text{for}\ i\in\{1,\cdots ,N\},
\label{DEFDYNAM}
\end{equation}
where $\alpha\in [0,1)$ and the parametrized one-dimensional maps $f_{\alpha,a}$ are defined by
\[
f_{\alpha,a}(p)=\alpha p+(1-\alpha) f_{a}(p).
\]
The various assumptions on $f_a$ and $g$ are presented and discussed below. 

For convenience, we may regard the iterations \eqref{DEFDYNAM} as the repeated action of some multidimensional map, {\sl viz.}\ equation \eqref{DEFDYNAM} implicitly defines a map $F:M_N\to M_N$ such that $(\mathbf p^{t+1},\mathbf a^{t+1})=F(\mathbf p^t,\mathbf a^t)$ for all $t\in\N$. In this context, we recall that the {\bf orbit} that starts at $t=0$ ({\sl ie.} first day) from a given initial condition $(\mathbf p^0,\mathbf a^0)$ is the sequence $\{(\mathbf p^t,\mathbf a^t)\}_{t\in\N}$ defined by $(\mathbf p^t,\mathbf a^t)=F^t(\mathbf p^0,\mathbf a^0)$.

\paragraph{Interpretation of the iteration rule \eqref{DEFDYNAM}.}
The expression in the first row of \eqref{DEFDYNAM} considers that a fraction $\alpha$ of the  population is loyal and systematically returns to each seller they purchased at the previous day, independently of the latter' attractiveness (term $\alpha p_i^t$). The loyalty parameter $\alpha$ is fixed once for all and remains unchanged throughout the paper. For the sake of notation, we shall not mention any explicit dependence on this parameter. 

The remaining $1-\alpha$ fraction of the population is sensitive to merchants' influence power. As the term $(1-\alpha) f_{a_i^{t+1}}(p_i^t)$ stipulates, the volume of influenceable buyers who purchase at merchant $i$ at $t+1$ not only depends on the attractiveness $a_i^{t+1}$ on that day, but is also updated according to some contagion effect ({\em ie.}\ the updated volume depends on $p_i^t$). Such contagion is driven by mimetic behaviours and this is expressed in the assumptions on $f_a$ below. In few words, these assumptions (together with the convex combination in the definition of $f_{\alpha,a}$) imply that the fraction of buyers at an attractive wholesaler must increase, with stronger effect when the attractiveness is higher (and conversely this volume must decrease, again with attractiveness dependent amplitude, in the case of a repulsive seller).

Notice that if the attractiveness $a_i^t=a_i>1$ did not depend on $t$, then, starting from any $p_i^0\in [0,1)$, the sequence $\{p_i^t\}_{t\in\N}$ obtained by iterating the map $f_{\alpha,a_i}$ would be increasing and converging to 1, indicating gradual convergence to maximal clientele for the wholesaler under consideration. Conversely, if $a_i<1$, the sequence would be decreasing and converging to 0, meaning buisness collapse in this case.

The expression in the second row of \eqref{DEFDYNAM} considers that merchants update their attractiveness according to the volumes of their proper clientele $p_i^t$ and of the mean clientele $\frac1{N}\sum_{j=1}^Np_j^{t}$. More precisely, together with condition \eqref{INEQG} below, this expression implies that we must have $a_i^{t+1}>a_i^t$ when the difference of these volumes is negative ({\sl viz.}\ when the proper clientele is smaller than the mean attendance) and $a_i^{t+1}<a_i^t$ when the difference is positive. 

These seller reactions regulate the dynamics through negative feedback. Indeed, the first row of \eqref{DEFDYNAM} and the assumptions on $f_a$ indicate that the clientele at the wholesaler with the highest attractiveness is likely to become the largest one (as long as that attractiveness remains the highest). However, the second row implies that from the day this happens, that wholesaler must reduce their attractiveness which hence may eventually cease to the be highest one, especially given that the opposite variation occurs at the wholesaler with lowest attractiveness. 

In other words, dominant and dominated sellers, as well as the volumes of their clientele, are prone to alternate in time, and such oscillatory behaviours is the hallmark of a negative regulatory feedback.\footnote{Negative feedback and subsequent possible oscillations due to the assumptions on the $f_a$ are more easily apprehended in the case $N=2$. Assume that we have $a_1^t<1<a_2^t$ for a number of consecutive values of $t$. Over this time interval, $p_1^t\searrow 0$ and $p_2^t\nearrow 1$. Depending on the interval length and on the values of $p_i^t$ at the beginning of this interval, we may eventually have $p_1^t<p_2^t$. This in turn implies $a_1^t\nearrow$ and $a_2^t\searrow$. If this trend also lasts sufficiently long, then we may eventually have $a_2^t<1<a_1^t$ and the whole process may repeat with the role of 1 and 2 exchanged.} Of note, properties of multidimensional dynamical systems implying such negative regulation have been investigated in the literature. In particular, the presence of a negative feedback loop in the underlying regulation graph has been identified as a necessary condition for a stationary state/periodic orbit to be stable \cite{G98,KST07,PMO95}. 

The goal of the analysis presented below is to investigate the long term behaviours as $t\to +\infty$ and in particular, to determine whether the outlined oscillations perdure or if they can be eventually damped so that the system converges to a stationary state. 

\paragraph{Assumptions on the maps $f_a$.} The parametrized maps $f_a$ in the contagion term of the first row of \eqref{DEFDYNAM} (which are all defined from $[0,1]$ into itself) must be such that (see Fig.\ \ref{GRAPHFA} for an illustration)
\begin{itemize}
\item $f_a(p)>p$ for every $p\in [0,1)$ and $f_a(1)=1$ when $a>1$, 
\item $f_a(p)<p$ for every $p\in (0,1]$ and $f_a(0)=0$ when $a<1$ (and $f_1=\text{Id}$).
\item $f_{a_1}(p)<f_{a_2}(p)$ for every $a_1<a_2\in \R_\ast^+$ and $p\in (0,1)$.
\end{itemize}
The first assumption here implies that, for every $a>1$, the point $p=1$ is a globally attracting fixed point of the map $f_a$. The second assumption implies a similar property for $p=0$ when $a<1$. The third assumption implies that the corresponding convergence rates depend monotonically on $a$. 

For the sake of analysis and simplicity, the following technical assumptions are also required.
\begin{itemize}
\item For every $a\in\R_\ast^+$, $f_a$ is an increasing $C^2$ map on $[0,1]$. Moreover, we have
\[
f'_a(0)=f'_{\frac1{a}}(1)=a,\ \forall a\in (0,1)\quad\text{and}\quad \sup_{a\in (0,1)}\max\{|f''_a(0)|,|f''_\frac1{a}(1)|\}<+\infty.
\]
\end{itemize}
The role of the assumption on $f_a'(0)$ and $f_\frac1{a}(1)$ is to ensure that, when $a<1$, the fixed point $p=0$ is exponentially stable, and similarly for $p=1$ when $a>1$. As long as $f_a'(0),f_\frac1{a}(1)<1$,  the exact values of $f_a'(0)$ and $f_\frac1{a}(1)$ do not matter. The choice $f'_a(0)=f'_{\frac1{a}}(1)=a$ is merely a matter of parametrization.
\begin{itemize}
\item For every $x\in [0,1]$, the map $a\mapsto f_a(x)$ is continuous in $\R_\ast^+$. The map $a\mapsto \|f'_a\|_\infty:=\max_{x\in [0,1]}f'_a(x)$ is continuous  in $\R_\ast^+$.
\end{itemize}
In particular, all the conditions above hold for the following simple quadratic example (see fig.\ \ref{GRAPHFA} for illustrations)\footnote{Instead of $0.9$ any constant in $(0,1)$ could have been chosen in this example. We believe that the closer this constant is to 1, the more elaborated is the transient phenomenology of the system \eqref{DEFDYNAM}.}
\begin{equation}
f_a(x)=\left\{\begin{array}{ccl}
a x + 0.9(1 - a) x^2&\text{if}&a\leq 1\\
1-\frac1{a}(1-x)-0.9(1-\frac1{a})(1-x)^2&\text{if}&a\geq 1
\end{array}\right.
\label{EXAMPLFA}
\end{equation}
\begin{figure}[ht]
\begin{center}
\includegraphics*[width=60mm]{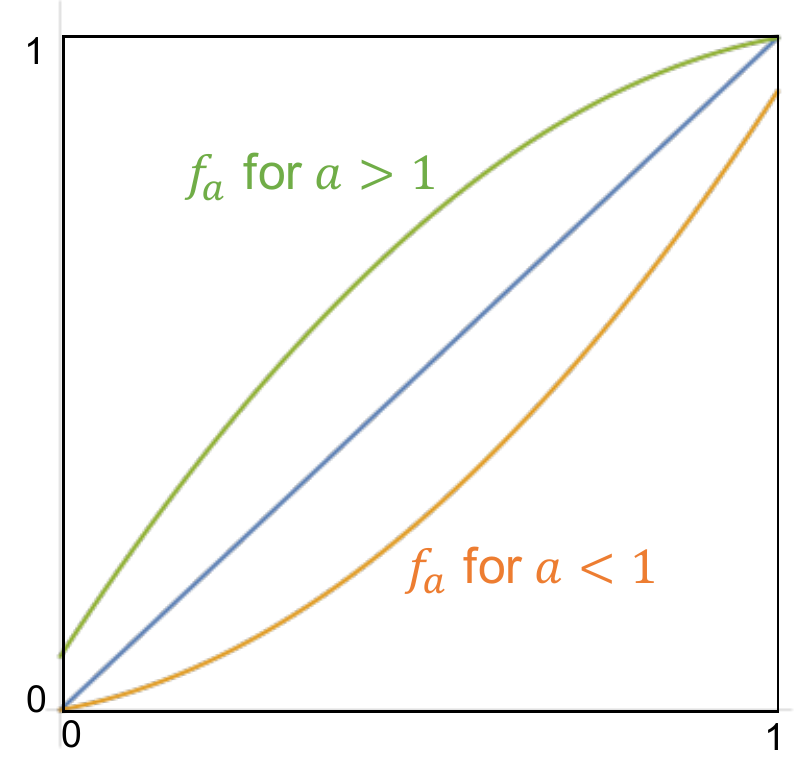}
\end{center}
\caption{Illustration of graphs of the parametrized functions $f_a$ defined by \eqref{EXAMPLFA}.}
\label{GRAPHFA}
\end{figure}

\paragraph{Assumptions on the map $g$.} 
 The map $g:[0,1]^2\to \R^+_\ast$ must be a continuous map such that $g(p,p)=1$ for all $p\in [0,1]$ and
\begin{equation}
(g(p,q)-1)(q-p)>0,\ \forall p\neq q\in [0,1].
\label{INEQG}
\end{equation}
This inequality is a condensed way to simultaneously impose that $g(p,q)>1$ when $q>p$ and $g(p,q)<1$ when $q<p$. 

Additional conditions will be imposed below - namely conditions \eqref{UPBOUNDG} and \eqref{CONCAV} - in order to control the asymptotic dynamics of the system \eqref{DEFDYNAM}, in particular the fixed points' local stability and the global features of the dynamics. 

Strictly speaking, we do not need to exclude that $g$ vanishes on $[0,1]^2$. It suffices that $g$ remains positive in all orbits of \eqref{DEFDYNAM}, namely that we have
\begin{equation}
g(p_i,\frac1{N}\sum_{i=1}^Np_i)>0, \ \forall \mathbf p\in [0,1]^N,\ i\in\{1,\cdots ,N\}\ \text{and}\ N\geq 2.
\label{POSITIVEG}
\end{equation}
This weaker assumption allows one to consider the following linear example
\begin{equation}
g(p,q)=1+q-p.
\label{LINEAREXAMP}
\end{equation}

\section{Main results}\label{S-MAINRES}
\subsection{Presentation of the results}
As it is standard in dynamical systems, the analysis of the system \eqref{DEFDYNAM} aims at describing the asymptotic behaviour as $t\to +\infty$ of every orbit, leaving out exceptional orbits with unstable or irrelevant behaviour. Given the negative nature of the regulatory process involved here, emphasis will be put on identifying those conditions that ensure that the oscillatory behaviours eventually cease, and convergence to a fixed point always holds (again leaving out exceptional orbits). 

The analysis begins with preliminary considerations that provide a basis for the subsequent presentation of the results. In particular, an inversion symmetry is identified that allows us to first consider an asymmetric set of initial conditions and then to extend the conclusions to their symmetric counterpart. 
Moreover, all fixed points are determined, and as potential attractors of the dynamics, focus will be made on those continuous families of fixed points with either $\mathbf p=\mathbf 0$ or $\mathbf p=1$, and which are parametrized by $\mathbf a$. 

The first result (Proposition \ref{STABP1}) is about local stability of these families of fixed points. More precisely, the statement claims that, provided that merchants' reactivity to clientele variations is bounded (condition \eqref{UPBOUNDG}), the continuum of fixed points with $\mathbf p=\mathbf 0$ is locally asymptotically stable, namely it attracts the orbit of every initial condition in its neighbourhood. Thanks to the inversion symmetry, the continuum with $\mathbf p=\mathbf 1$ is also locally asymptotically stable provided that a symmetric analogue to \eqref{UPBOUNDG} holds.  In other words, condition \eqref{UPBOUNDG} (and the symmetric one) imply that the individual merchants' attraction/repulsion power dominates feedback regulation when near enough to the fixed points. Oscillations in the orbits must be eventually damped and asymptotic convergence takes place.

The second result (Proposition \ref{INSTAB}) puts emphasis on  the importance of condition \eqref{UPBOUNDG} for the stability of the continuum of fixed points with $\mathbf p=\mathbf 0$ (NB: By symmetry, a similar discussion holds for the fixed points with $\mathbf p=1$). More precisely, this statement provides a counter-example, namely a map $g$ not satisfying  \eqref{UPBOUNDG} and for which this continuum of  fixed points  is unstable. In that counter-example, merchants have unbounded reactivity (meaning that they are hyper reactive) when their clientele volume is very small, a feature that saves them from total collapse.

Leaving the counter-example aside, and assuming again that condition \eqref{UPBOUNDG} holds, the third and last result, Theorem \ref{MAINRES}, which is the main result of the paper, completes the local stability result of Proposition \ref{STABP1} by a global stability one. It {\sl a priori} not clear that \eqref{UPBOUNDG}, which is a local condition, suffices to control the (oscillatory) dynamics when starting away from $\mathbf p=\mathbf 0$. 

Accordingly, Theorem \ref{MAINRES} identifies an additional condition on $g$, condition \eqref{CONCAV}, that curbs the mean attractiveness growth rate, and which, together with \eqref{UPBOUNDG}, ensures global convergence to a fixed point. More precisely, the statement claims convergence for every orbit except perhaps when an attractiveness $a_i^t$ approaches 1 arbitrarily close.\footnote{In this case, we cojecture that the variations should be become arbitrarily slow although they may not vanish. } We believe that such exceptions, when some wholesaler are (infinitely often close to being) neither attractive nor repulsive, are non-generic, even though they cannot be totally excluded.\footnote{In particular, such exceptions may happen when in the boundaries of the basins of attractions of the fixed points $\mathbf 0$ and $\mathbf 1$.} In any case, our study concludes that other than these exceptions, the population individual stability always dominates the feedback regulation to impose convergence in this system, provided that the wholesalers' reactivity remains bounded and that the attractiveness mean grow rate is not positive.

\subsection{Preliminary properties}
\paragraph{Well-defined dynamics.} 
The convex combination in the first equation of \eqref{DEFDYNAM}, together with the property $f_a([0,1])\subset [0,1]$, and the inequality \eqref{POSITIVEG} imply that for every initial condition $(\mathbf p^0,\mathbf a^0)\in M_N$, the subsequent orbit $\{(\mathbf p^t,\mathbf a^t)\}_{t\in\N}$ is well-defined, and we have $(\mathbf p^t,\mathbf a^t)\in M_N$ for all $t\in\N$. 
  
\paragraph{Identification of the fixed points.} 
Equation \eqref{DEFDYNAM} implies that the $\mathbf p$-coordinates of every fixed point $F(\mathbf p,\mathbf a)=(\mathbf p,\mathbf a)$ must satisfy
\[
g(p_i,\frac1{N}\sum_{i=1}^Np_i)=1, \ \forall i\in\{1,\cdots ,N\}.
\]
The condition \eqref{INEQG} then imposes that these $\mathbf p$-coordinates must be all equal. Moreover, $\alpha\neq 1$ imposes that, for each $i$, the coordinate $p_i$ must be a fixed point of the map $f_{a_i}$. Given the conditions on the maps $f_a$, this in turn implies that we must have $p_i=0$ if $a_i<1$, $p_i=1$ if $a_i>1$ and $p_i$ is arbitrary in $[0,1]$ if $a_i=1$. 

Let $\mathbf 0=(0,\cdots ,0)$ and $\mathbf 1=(1,\cdots ,1)$. It results from the reasoning here that the fixed points of $F$ are of the following form
\begin{itemize}
\item $(\mathbf 0,\mathbf a)$ for any $\mathbf a\in (0,1]^N$,
\item $(\mathbf 1,\mathbf a)$ for any $\mathbf a\in [1,+\infty)^N$,
\item $(\mathbf p,\mathbf 1)$ for any $\mathbf p\in [0,1]^N$.\footnote{Notice that the point $(\mathbf 0,\mathbf 1)$ also belongs to the first family and $(\mathbf 1,\mathbf 1)$ also belongs to the second one.}
\end{itemize}
The fixed points with neutral attractiveness $\mathbf a=\mathbf 1$ are a consequence of the mathematical properties of the family $f_a$ (continuity with respect to $a$ in particular) and have little significance for market modelling, also because the $\mathbf p$-coordinates are indifferent. In fact, these points are particular homogeneous points whose dynamics reduces to a one-dimensional dynamical system (synchronized dynamics, see below) and they are unstable fixed points of the synchronized dynamics. 

On the other hand, and despite that the $\mathbf a$-coordinates are indifferent,\footnote{Such indifference is a consequence of the fact that once the attractiveness exceeds 1 (or lies below 1), no selection mechanism for that quantity exists in the model \eqref{DEFDYNAM}.} fixed points with $\mathbf p=\mathbf 0,\mathbf 1$ and all $\mathbf a$-coordinates distinct from 1 play a central role in the dynamics because they are natural candidates for attractive fixed points (Recall the assumptions above on the one-dimensional maps $f_a$ and the subsequent exponential stability of the fixed points 0 and 1). Accordingly, the aim of the analysis below is to identify additional conditions on $g$ that ensure that these families of points are first locally stable, and then globally stable, in order to evacuate permanent oscillations (except perhaps for exceptional orbits). 

Of note, as they can be directly interpreted, these fixed points (total collapse of each wholesaler when $\mathbf p=\mathbf 0$, all buyers purchasing at all merchants when $\mathbf p=\mathbf 1$) do not appear to be fully realistic. This limitation is another consequence of the simplicity of the dynamics, in particular of the choice of the fixed points of the maps $f_a$ themselves. However, when specifying these maps, one could adjust this choice (possibly depending on $N$), in order to curb the minimal and maximal clientele volumes at each merchant, and hence to curb the fixed points coordinates. While such adjustment is highly relevant for modelling, it does not affect the asymptotic behaviours, which can be deduced from the results here by applying a rescaling of the variables.

\paragraph{Synchronized dynamics.}
As a particular instance of coupled map system with mean-field coupling \cite{CF05}, the dynamics \eqref{DEFDYNAM} commutes with the simultaneous permutations of the $\mathbf p$- and $\mathbf a$-coordinates, {\sl viz.}\ if the sequence $\{(\mathbf p^t,\mathbf a^t)\}_{t\in\N}$ is an orbit of \eqref{DEFDYNAM}, then  for every permutation $\pi$ of $\{1,\cdots ,N\}$, the sequence $\{(\pi \mathbf p^t,\pi \mathbf a^t)\}_{t\in\N}$ where 
\[
\pi \mathbf p=(p_{\pi(1)},\cdots,p_{\pi(N)})\quad\text{and}\quad \pi\mathbf a=(a_{\pi(1)},\cdots ,a_{\pi(N)}),
\]
is also an orbit of this system.

This symmetry implies in particular that the subset of $M_N$ of {\bf homogeneous} points, {\sl viz.}\ when $\mathbf p=(p,\cdots ,p)$ and $\mathbf a=(a,\cdots ,a)$, is invariant under the action of $F$. Moreover, the dynamics in this set - the so-called synchronized dynamics - reduces to the iterations of the one-dimensional map $f_{\alpha,a}$ (since $a$ remains constant). The synchronized dynamics is very simple; for every $a\neq 1$, every orbit asymptotically converges to 0 (when $a<1$) or to 1 (when $a>1$). In the original phase space, this corresponds to the convergence to one of the homogeneous fixed points $(\mathbf 0,\mathbf a)$ or $(\mathbf 1,\mathbf a)$ where again $\mathbf a=(a,\cdots ,a)$. For $a=1$, we have $f_{\alpha, 1}=\text{Id}$, thus every orbit is stationary.

Clearly, the map $F$ is invertible, {\sl ie.}\ $(\mathbf p^t,\mathbf a^t)$ in \eqref{DEFDYNAM} can be uniquely determined by $(\mathbf p^{t+1},\mathbf a^{t+1})$. Moreover, its inverse also commutes with the simultaneous permutations of the $\mathbf p$- and $\mathbf a$-coordinates. As a consequence, no orbit can become synchronized in finite time if its initial condition is not homogeneous.\footnote{Indeed, if, at any given $t\in\N$, the coordinates of $\mathbf p^t$ and $\mathbf a^t$ do not depend on $i$, then the same property must hold for $(\mathbf p^0,\mathbf a^0)$.} 

Accordingly, the description of the synchronized dynamics above cannot apply to non-synchronized orbits. A proper analysis is required in order to determine asymptotic behaviours in this case. Of note, in addition to their application to the proofs of the main results, some of the technical statements in Section \ref{S-PROOF} below provide interesting information about the asymptotic behaviours, see especially section \ref{P-MAINRES}.

\paragraph{Inversion symmetry.} Here, we identify an involution acting in the set of maps $F$ associated with the dynamics \eqref{DEFDYNAM} and which commutes with every such map. The subsequent benefits for the analysis will be also given. 

The involution is based on the following similar transformations for the maps $f_a$ and $g$. 
\begin{itemize}
\item Given a parametrized map $f_a:[0,1]\to [0,1]$, let $\bar f_a$ be defined by 
\[
\bar f_a(x)=1-f_{\frac1{a}}(1-x),\ \forall x\in [0,1],a\in\R_\ast^+.
\]
Clearly, when the family $\{f_a\}_{a\in\R_\ast^+}$ satisfies all assumptions in Section \ref{S-DEF}, the family $\{\bar f_a\}_{a\in\R_\ast^+}$ also satisfies these assumptions. Notice also that we have $\bar f_a=f_a$ for the map defined by \eqref{EXAMPLFA}.
\item Given a map $g:[0,1]^2\to \R_\ast^+$, let $S g$ be defined by 
\begin{equation}
S g(p,q)=\frac1{g(1-p,1-q)},\ \forall (p,q)\in [0,1]^2.
\label{SYMG}
\end{equation}
Again, if $g$ satisfies all assumptions in Section \ref{S-DEF}, then so does the image $Sg$. 
\end{itemize}
Now, there is a one-to-one correspondence between the orbits of \eqref{DEFDYNAM} and those of the dynamical system that results when substituting $f_a$ by $\bar f_a$ and $g$ by $S g$. Namely, the sequence $\{(\mathbf p^t,\mathbf a^t)\}_{t\in\N}$ is an orbit of the former iff $\{(1-\mathbf p^t,\frac1{\mathbf a^t})\}_{t\in\N}$, where 
\[
1-\mathbf p=(1-p_1,\cdots ,1-p_N)\quad \text{and}\quad \frac1{\mathbf a}=(\frac1{a_1},\cdots,\frac1{a_N}),
\]
is an orbit of the latter. 

As the benefits for the analysis are concerned, notice firstly that the stability statements to be given in the sequel depend on some conditions on the map $g$ only. Therefore, for a given statement, if the map $Sg$ also satisfies the same condition, the conclusion for an orbit $\{(\mathbf p^t,\mathbf a^t)\}_{t\in\N}$ immediately extends to the orbit $\{(1-\mathbf p^t,\frac1{\mathbf a^t})\}_{t\in\N}$ of the same system (by applying the statement to the orbit $\{(\mathbf p^t,\mathbf a^t)\}_{t\in\N}$ of the system generated with $\bar f_a$ and $S g$). Specific consequences of this property on fixed point stability will be mentioned below.

\subsection{Local asymptotic stability}
The first equation in \eqref{DEFDYNAM} and the basic properties of the maps $f_a$ imply that in any orbit for which there exists $T\in\N$ such that
\begin{equation}
\sup_{t> T}\max_{i\in \{1,\cdots ,N\}} a_i^t< 1\quad\text{or}\quad \inf_{t> T}\min_{i\in \{1,\cdots ,N\}} a_i^t> 1,
\label{EVENTUALCONTROL}
\end{equation}
all coordinates $p_i^t$ must be monotone in $t$ for $t\geq T$, and hence they must asymptotically converge (to 0 in the first case, to 1 otherwise). 
However, that such behaviours exist is not clear because as outlined before, at every iteration, some of the $\mathbf a$-coordinates must increase and other(s) must decrease (unless the $\mathbf p$-coordinates are all equal). Therefore, in order to make sure that \eqref{EVENTUALCONTROL}  eventually holds, one needs to show that these variations asymptotically vanish and that $\mathbf a^t$ converges (to a point whose coordinates are all distinct from 1).

Our first result claims asymptotic convergence to the continuum of fixed points with $\mathbf p=\mathbf 0$ when starting initially close to that set, provided that the merchants' reactivity to clientele changes, that is the amplitude of the attractiveness variations, remains bounded. As mentioned above, the asymptotic value of $\mathbf a$ is itself hardly predictable from the initial condition (due to degeneracy inside the continuum), because no mechanism for attractiveness selection has been inserted into the dynamics. 
\begin{Pro}
Assume the existence of $K\in \R_\ast^+$ such that\footnote{The value $\frac12$ here is arbitrary. Any value in $(0,1)$ can be chosen.}
\begin{equation}
|g(p,q)-1|\leq K\max\{p,q\},\quad \forall (p,q)\in (0,\frac12)^2.
\label{UPBOUNDG}
\end{equation}
Then for every $\mathbf a^0\in (0,1)^N$, there exists $\epsilon>0$ such that for every $\mathbf p^0\in [0,\epsilon)^N$, the subsequent orbit of \eqref{DEFDYNAM} satisfies
\[
\sup_{t\in\N}\max_{i\in\{1,\cdots ,N\}}a_i^t<1\quad \text{and}\quad \lim_{t\to +\infty}\mathbf a^t\ \text{exists},
\] 
and hence
\[
\mathbf p^t\in [0,\epsilon)^N,\ \forall t\in\N\quad \text{and}\quad \lim_{t\to +\infty}\mathbf p^t=0.
\] 
\label{STABP1}
\end{Pro} 
In addition, an analogous result follows by symmetry for the continuum of fixed points with $\mathbf p=1$, when $S g$ defined in \eqref{SYMG} satisfies \eqref{UPBOUNDG}.  That is to say, for every $\mathbf a^0\in (1,+\infty)^N$, there exists $\epsilon>0$ such that for every $\mathbf p^0\in (1-\epsilon,1]^N$, the subsequent orbit of \eqref{DEFDYNAM} (still with $f_{\alpha,a}$ and $g$) satisfies
\[
\inf_{t\in\N}\min_{i\in\{1,\cdots ,N\}}a_i^t>1\quad \text{and}\quad\lim_{t\to +\infty}\mathbf a^t\ \text{exists},
\] 
and hence
\[
\mathbf p^t\in (1-\epsilon,1]^N,\ \forall t\in\N\quad \text{and}\quad \lim_{t\to +\infty}\mathbf p^t=1.
\] 
The proof of Proposition \ref{STABP1} is given in Section \ref{P-STABP1}. As an example, notice that  both the linear map $g$ defined by \eqref{LINEAREXAMP} and its image $Sg$ satisfy condition \eqref{UPBOUNDG}.
\begin{figure}[ht]
\begin{center}
\includegraphics*[width=140mm]{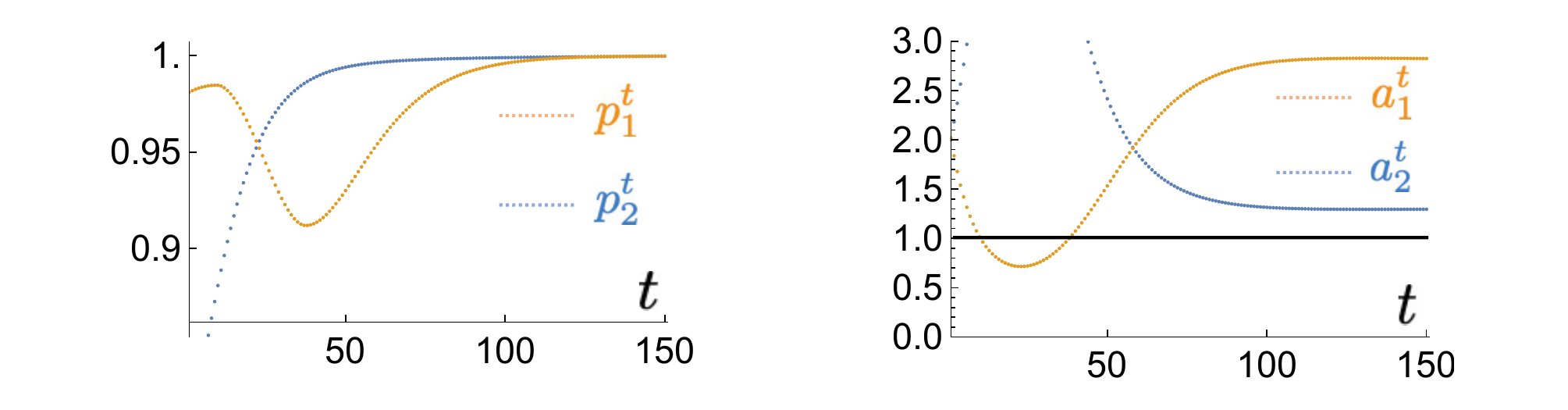}
\end{center}
\caption{Illustration for $N=2$ of the (local) asymptotic stability of the set of fixed points with $\mathbf p=\mathbf 1$. The pictures represent the $\mathbf p^t-$ and $\mathbf a^t-$coordinates time series of the orbit of \eqref{DEFDYNAM} issued from $(\mathbf p^0,\mathbf a^0)=(0.981,0.8,2.02,2)$, for $f_a$ as in \eqref{EXAMPLFA}, $\alpha=0.9$ and $g$ defined by \eqref{LINEAREXAMP}. Prompt convergence towards $(\mathbf 1,\mathbf a)$ (with both $a_i>1$) is evident, despite a short oscillatory transient of $p_1^t$ due $a_1^t$ passing briefly below 1. This example can be regarded as an illustration of the dynamics that results from introducing a concurrent seller (indexed by 2) to an existing isolated one (seller 1) whose initial clientele fraction and attractiveness coefficients are similar, although more favourable. The initial strong increase of $a_2^t$ reveals a rapid reaction from seller 2 in order to attract more customers and to bring the fraction $p_2^t$ above $p_1^t$. When that happens, the roles and the variations of the $a_i^t $ are exchanged and monotone convergence follows.}
\label{GROWTH}
\end{figure}

As market modelling is concerned, Proposition \ref{STABP1} can be interpreted as a robustness statement when close to vanishing or to maximal clientele. In particular, when $Sg$ satisfies \eqref{UPBOUNDG}, it shows that the introduction of an additional merchant into a market where each seller already attracts almost all buyers, does not affect the asymptotic behaviour of the  populations, provided that the newcomer also initially attracts a large part of the population. 

An illustration of the dynamics for $N=2$ when starting initially close to $\mathbf p=\mathbf 1$ with $a_1^0,a_2^0>1$, is given in Fig.\ \ref{GROWTH}. Notice that in this example, the initial condition does not satisfy the conditions of Proposition \ref{STABP1} because $a_1^t<1$ during some transient time interval (say from $t\sim 12$ to $t\sim 27$). This suggests that convergence to $\mathbf p=\mathbf 1$ may hold for a larger set of initial conditions than as prescribed by (the symmetric analogue of) Proposition \ref{STABP1}, when the subsequent orbit shows that some merchants temporary become repulsive.

\subsection{Role of condition \eqref{UPBOUNDG} on fixed points stability}\label{S-INSTAB}
Condition \eqref{UPBOUNDG} in Proposition \ref{STABP1} and its analogue for the map $Sg$ are essential for the local stability of the continua of fixed points with $\mathbf p=\mathbf 0,\mathbf 1$. Indeed, the following counter-example of the map $g:(0,1]\times [0,1]\to \R_\ast^+$ defined by
\begin{equation*}
g(p,q)=\frac{q}{p},
\end{equation*}
shows that instability can result when this condition fails.\footnote{Independently, on $(0,1]\times [0,1]$, this map satisfies all the assumptions imposed on $g$ in Section \ref{S-DEF}. In this case, the dynamics of \eqref{DEFDYNAM} is well-defined in $(0,1]^N\times (\R_\ast^+)^N$, which is invariant.}
\begin{Pro}
For every $\mathbf a^0\in (0,1)^N$ and every $\mathbf p\in (0,1)^N$ such that $(\mathbf p,\mathbf a^0)$ is not synchronized, there exist $T\in\N$ and $\Delta>0$ such that for every $\delta\in (0,\Delta)$ we have
\[
\max_{i\in\{1,\cdots ,N\}}a_{i}^{T}> 1,
\]
for the orbit issued from the initial condition $(\delta \mathbf p,\mathbf a^0)$ and generated by \eqref{DEFDYNAM} with $g(p,q)=\frac{q}{p}$.
\label{INSTAB}
\end{Pro} 
As before, an analogous conclusion holds by symmetry for the system generated by \eqref{DEFDYNAM} with $g(p,q)=\frac{1-q}{1-p}$. The proof of Proposition \ref{INSTAB} is given in Section \ref{P-INSTAB}.

As market modelling is concerned, Proposition \ref{INSTAB} can be interpreted as merchants' resilience against clientele collapse. Unlike when \eqref{UPBOUNDG} holds, reactive changes in the attractiveness coefficients can become sufficiently large to counteract population decay, and to make the continuum of points with $\mathbf p=\mathbf 0$ being a repeller, {\sl viz.}\ no matter how small the initial fractions $p_i^0$ are, there exists a time for which at least one wholesaler must be attractive, implying that the corresponding fraction must increase, at least for one iteration.

For completeness, we mention that numerical simulations of this example with $N=2$ (see Fig. \ref{DYNAMICSN2}) actually show that the instability is so strong that the two attractiveness $a_1^t$ and $a_2^t$ eventually exceed 1 and stay above that value. As discussed in the previous subsection, this implies convergence towards a fixed point with $\mathbf p=\mathbf 1$. Notice that the symmetric map $Sg$ satisfies \eqref{UPBOUNDG} in this example, implying that the corresponding continuum is locally asymptotically stable (statement symmetric to Proposition \ref{STABP1}). 
\begin{figure}[ht]
\begin{center}
\includegraphics*[width=140mm]{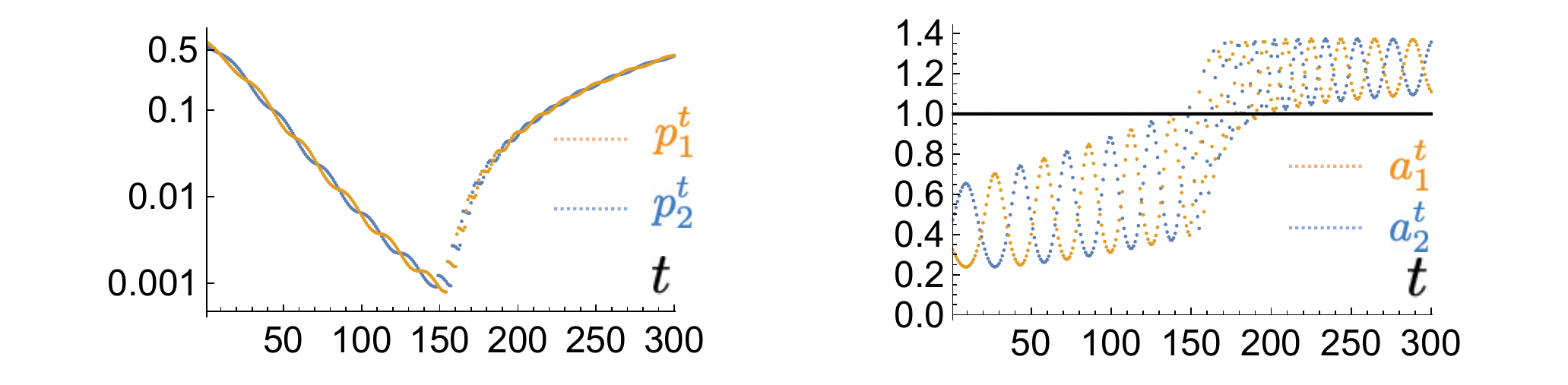}
\end{center}
\caption{Illustration for $N=2$ of the instability of the set of fixed points with $\mathbf p=\mathbf 0$ for $g(p,q)=\frac{q}{p}$ which does not satisfies the condition \eqref{UPBOUNDG} ((NB: Notice however that $Sg$ satisfies \eqref{UPBOUNDG}). The pictures represent the $\mathbf p^t-$ and $\mathbf a^t-$coordinates time series of the orbit issued from $(\mathbf p^0,\mathbf a^0)=(0.546,0.616,0.473,0.324)$, for $f_a$ as in \eqref{EXAMPLFA} and $\alpha=0.9$. Notice the logarithmic scale on the left picture. The figure clearly indicates that a transient oscillating instability of the $a_i^t$ occurs that eventually brings the two coefficients above 1. Accordingly, after a first phase of exponential decay to 0, both coordinates $p_i^t$ start to grow and eventually converge to 1 (as expected when both $a_i^t>1$).}
\label{DYNAMICSN2}
\end{figure}

\subsection{Asymptotic behaviour of orbits}
Back to assuming that $g$ satisfies \eqref{UPBOUNDG}, it remains to investigate the dynamics when starting away from the continua of fixed points with $\mathbf p=\mathbf 0$ and with $\mathbf p=\mathbf 1$. This question is motived in particular by the example in Fig.\  \ref{GROWTH} which suggests that asymptotic convergence to a fixed point may (also) hold for initial conditions no close enough to these continua, and thus not satisfying the assumptions of Proposition \ref{STABP1}. 

However, because of the merchant's negative feedback, oscillatory behaviours may occur over arbitrary large time intervals if the coefficients $a_i^t$ cross 1 arbitrarily many times. In particular, this may be the case when close to the boundary of the fixed points' basins. Furthermore, we do not exclude that, exceptionally, oscillations last forever (when in addition to infinitely many crossings of 1, the coefficients $a_i$ approach 1). 

The central result of this paper integrates all these considerations at once, and states that every orbit of \eqref{DEFDYNAM} must asymptotically converge, and hence that oscillations must be evnetually damped, unless the attractiveness coefficients do not remain away from 1. In addition to \eqref{UPBOUNDG}, this result requires that $g$ satisfies some weak form of concavity, more precisely that its mean value in phase space cannot exceed 1.
\begin{Thm}
In addition to satisfying condition \eqref{UPBOUNDG}, assume that $g$ satisfies that following concavity inequality
\begin{equation}
\frac1{N}\sum_{i=1}^Ng(p_i,\frac1{N}\sum_{i=1}^Np_i)\leq 1,\ \forall\mathbf p\in (0,1)^N.
\label{CONCAV}
\end{equation}
Then, for every orbit $\{(\mathbf p^t,\mathbf a^t)\}_{t\in\N}$ of \eqref{DEFDYNAM}, we have
\[
\sup_{t\in\N}\max_{i\in\{1,\cdots ,N\}}a_i^t<+\infty.
\]
Moreover, every orbit that satisfies  
\[
\liminf_{t\to +\infty}\min_{i\in\{1,\cdots ,N\}}|a_i^t-1|>0.
\]
must asymptotically converge to a fixed point or, possibly, to a ghost fixed point\footnote{A ghost fixed point is a point of the form $(\mathbf 0,\mathbf a)$ where $a_i=0$ for some $\in\{1,\cdots ,N\}$. The term {\sl ghost} comes from the fact that such points belong to $\overline{M_N}\setminus M_N$. In other words, ghost fixed points can be regarded as fixed points of the extension of the map $F$ to the closure $\overline{M_N}$.}.
\label{MAINRES}
\end{Thm}
Theorem \ref{MAINRES} follows from a more technical statement, namely Proposition \ref{LASTP}, which is stated and proved in Section \ref{P-MAINRES}. 

The fact that non-convergence to a (ghost) fixed point requires that some attractiveness coefficient(s) must accumulate at 1, namely
 \[
\liminf_{t\to +\infty}\min_{i\in\{1,\cdots ,N\}}|a_i^t-1|=0,
\]
suggests that this may only occur for exceptional initial conditions.\footnote{How exceptional such initial conditions are can be apprehended by considering the one-dimensional synchronized dynamics for which $a_i^t$ and $p_i^t$ do not depend on $i$. In this case, we have $\lim_{t\to +\infty} a^t=1$ iff $a^0=1$.} No such behaviour has ever been observed in simulations.

For instance, the linear map defined in \eqref{LINEAREXAMP} satisfies both conditions \eqref{UPBOUNDG} and \eqref{CONCAV}. As before, a similar statement holds in the case that $Sg$ satisfies \eqref{UPBOUNDG} and \eqref{CONCAV}. In particular, when $g(p,q)=\frac{q}{p}$, the map $Sg$ also satisfies these two conditions (see in particular Section \ref{P-INSTAB}). Up to a minor adjustment (see footnote 17), the proof of Theorem \ref{MAINRES} also applies in this example. 

As application to modelling is concerned, Theorem \ref{MAINRES} states that, possibly after a transient of variable duration, the market must eventually stabilize to a stationary functioning mode, unless some sellers eventually become repeatedly indifferent, an event that we believe to be highly unlikely. 

Moreover, as illustrated in Fig.\ \ref{GLOBALBEHAV}, transient oscillatory phases can be particularly long when starting initially close to the boundary of a fixed point's basin. Hence, stabilization might be hardly observed in practice, and one may only see oscillations, if the observation window is not long enough. Besides, it is also interesting to note that, again when close to a basin's boundary, the asymptotic state, whether $\mathbf p=\mathbf 0$ or $\mathbf p=\mathbf 1$, is sensitive to the initial condition. In other words, the asymptotic functioning mode might be hardly predictable from the first days of the market in this case. 
\begin{figure}[ht]
\begin{center}
\includegraphics*[width=140mm]{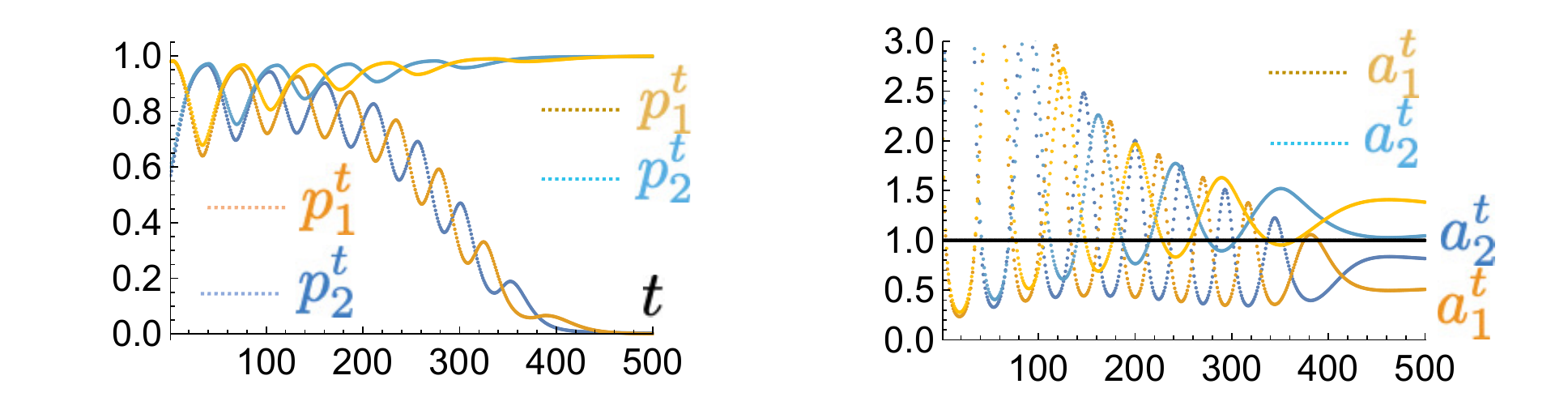}
\includegraphics*[width=140mm]{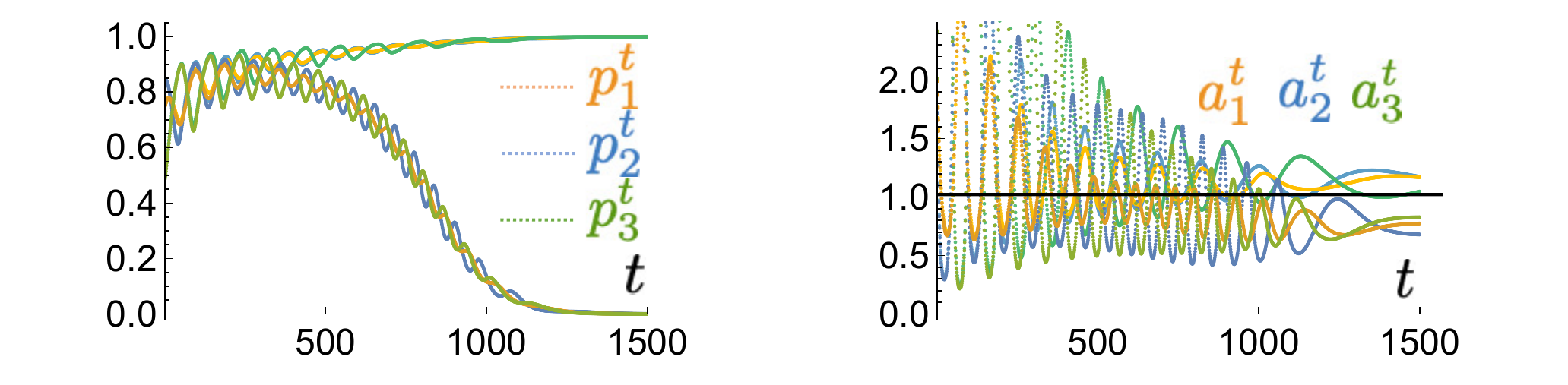}
\end{center}
\caption{Illustration for $N=2$ (top) and $N=3$ (bottom) of the transient oscillatory behaviours when initially close to the boundary of the fixed points' basins. Similar time series as in previous figures but now for two orbits issued from close initial conditions across the boundary of the fixed points' basins: For $N=2$, $a_2^0=0.57$ for the orbit that converges to $\mathbf p=\mathbf 0$ (darker colours) and $a_2^0=0.6$ for the other orbit (lighter colours). The other initial coordinates and the parameters are as in Fig.\ \ref{GROWTH}. For $N=3$, $p_3^0=0.487$ for the orbit that converges to $\mathbf p=\mathbf 0$ and $p_3^0=0.497$ for the other orbit. All other coordinates and parameters are equal. As application to market modelling is concerned, the figure indicates that the outcome that results from introducing a new competitor in the market - either extinction of the buyer population or convergence to maximal population - can be highly sensitive to small variations in the initial coordinates and may only emerge after a transient of substantial duration.}
\label{GLOBALBEHAV}
\end{figure}

Condition \eqref{CONCAV} can be interpreted as a kind of weak global dissipation of the attractivenesses. We know that, in non-synchronized orbits, some of the attractivenesses must grow and other must decrease between day $t$ and day $t+1$. Condition \eqref{CONCAV} ensures that, on average, decay or stagnation dominates. We suspect that Theorem \ref{MAINRES}, and in particular the first claim that the coefficients $\mathbf a^t$ remain bounded in every orbit, may no longer hold when \eqref{CONCAV} is not satisfied. 

\section{Stability analysis and proofs of the statements}\label{S-PROOF}
\subsection{Proof of Proposition \ref{STABP1}}\label{P-STABP1}
We only prove the first statement of the Proposition, since the one that involves $Sg$ immediately follows by applying the global symmetries, using that the stability of the continuum of fixed points with $\mathbf p=\mathbf 1$ for the original $f_a$ and $g$ is equivalent to the one with $\mathbf p=\mathbf 0$ for the system with $\bar f_a$ and $Sg$. 

The proof starts with the following crucial bootstrap statement that will also serve for future purposes.
\begin{Claim}
Given $a\in (0,1)$ and $\beta\in (\alpha+a(1-\alpha),1)$, there exists $\zeta_{a,\beta}>0$ such that for every $t\in\N$ such that 
\[
\left\{\begin{array}{l}
p_i^t\leq \frac{(1-\beta)\zeta}{K(a-\zeta)}\\
a_i^t\leq a-\zeta
\end{array}\right.\quad \text{for}\ i\in\{1,\cdots ,N\},
\]
where $K$ is given by the condition \eqref{UPBOUNDG} and $\zeta\in (0,\zeta_{a,\beta})$, we have 
\[
\left\{\begin{array}{l}
p_i^{t+1}\leq \frac{(1-\beta)\beta\zeta}{K(a-\beta\zeta)}\\
a_i^{t+1}\leq a-\beta \zeta
\end{array}\right.\quad \text{for}\ i\in\{1,\cdots ,N\}.
\]
\label{LOCAL}
\end{Claim}
\noindent
{\sl Proof of the Claim.} Let $\zeta_{a,\beta}>0$ be such that $\frac{(1-\beta)\zeta}{K(a-\zeta)}<\frac12$ for all $\zeta\in (0,\zeta_{a,\beta})$. The definition of $a_i^{t+1}$ and inequality \eqref{UPBOUNDG} imply that we have 
\[
a_i^{t+1}\leq a_i^t\left(1+\frac{(1-\beta)\zeta}{a-\zeta}\right)\leq a-\beta\zeta
\]
Moreover, given the choice of $\beta$, let $\epsilon>0$ be sufficiently small so that $\alpha+(1-\alpha)(a+\epsilon)\leq (1-\epsilon)\beta$. By continuity of the derivative $f'_a$ and the fact that $f'_a(0)=a$, let $\zeta_{a,\beta}>0$ be even smaller if necessary so that we have 
\[
f_{a_i^{t+1}}(x)\leq f_{a}(x)\leq (a+\epsilon)x,
\]
when $x\leq \frac{(1-\beta)\zeta_{a,\beta}}{K(a-\zeta_{a,\beta})}$. Let again $\zeta_{a,\beta}$ be even smaller if necessary so that we also have 
\[
1-\epsilon\leq \frac{a-\zeta}{a-\beta\zeta},\quad\text{for all}\ \zeta\in (0,\zeta_{a,\beta}).
\]
Altogether, we then have for every $p_i^t\leq \frac{(1-\beta)\zeta}{K(a-\zeta)}$ with $\zeta\in (0,\zeta_{a,\beta})$ 
\[
p_i^{t+1}\leq \frac{(1-\epsilon)(1-\beta)\beta\zeta}{K(a-\zeta)}\leq \frac{(1-\beta)\beta\zeta}{K(a-\beta\zeta)},
\]
as desired. \hfill $\Box$

\paragraph{Proof of Proposition \ref{STABP1}.} Given $\mathbf a^0\in (0,1)^N$, let $a\in (\max_i a_i^0,1)$, $\beta\in (\alpha+a(1-\alpha),1)$ and 
\[
\zeta< \min \left\{\zeta_{a,\beta},a-\max_i a_i^0\right\}.
\]
Then for $\max_i p^0_i<\frac{(1-\beta)\zeta}{K(a-\zeta)}$, Claim \ref{LOCAL} implies that we have
\[
\sup_{t\in\N}\max_{i\in\{1,\cdots ,N\}}a_i^t\leq a.
\]
Thanks to the properties of the map $f_a$ in the neighbourhood of $0$, this inequality not only implies that all $p_i^t$ tend to 0 but that this convergence is exponential, namely there exists $C\in\R^+$ and $b\in (0,1)$ such that 
\[
\max_{i\in\{1,\cdots ,N\}} p_i^t\leq Cb^t,\quad \forall t\in\N.
\]
According to the inequality \eqref{UPBOUNDG}, this yields the inequality
\[
\max_{i\in\{1,\cdots ,N\}}|a_i^{t+1}-a_i^t|\leq a K Cb^t\quad \forall t\in\N,
\]
which in turns implies that all sequences $\{a_i^t\}_{t\in\N}$ are Cauchy sequences, and hence ensures the existence of the limits $\lim_{t\to\infty}a_i^t$. \hfill $\Box$

\subsection{Proof of Proposition \ref{INSTAB}}\label{P-INSTAB}
The proof relies on the following preliminary statement.
\begin{Claim}
Assume that $\mathbf p^0\in (0,1)^N$ and that $(\mathbf p^0,\mathbf a^0)$ is not synchronized. Then, in the orbit issued from $(\mathbf p^0,\mathbf a^0)$,  there are infinitely many instances of $t\in\N$ such that the coordinates of $\mathbf p^t$ are not all equal.
\label{NONSYNC}
\end{Claim}
\noindent
{\sl Proof of the Claim.} Assume that all coordinates of $\mathbf p^t$ are equal. Then we must have $\mathbf a^{t+1}=\mathbf a^t$. Moreover, the coordinates of $\mathbf a^t$ cannot be all equal. Otherwise, $(\mathbf p^t,\mathbf a^t)$ would have to be synchronized, and hence $(\mathbf p^0,\mathbf a^0)$ would have to be synchronized too as we showed in Section \ref{S-MAINRES}. 

In addition, that some of the coordinates of $\mathbf a^t=\mathbf a^{t+1}$ are distinct, strict monotonicity of $a\mapsto f_{\alpha,a}(p)$ for every $p\in (0,1)$ and the fact that we must have $\mathbf p^t\not\in \{\mathbf 0,\mathbf 1\}$,\footnote{Indeed, by assumptions on the maps $f_a$ and since $\alpha \in [0,1)$, the only way that $\mathbf p^t\in \{\mathbf 0,\mathbf 1\}$ for $t>0$ is that the same condition holds for $t=0$, which is impossible given the current assumption on $\mathbf p^0$.} imply that some of the coordinates of $\mathbf p^{t+1}$ must be distinct. The Claim then easily follows.\hfill $\Box$
\medskip

For the proof of the Proposition, consider the product $\pi^t:=\prod_{i=1}^Na_i^t$. Expression \eqref{DEFDYNAM} implies that we have
\begin{equation*}
\frac{\pi^{t+1}}{\pi^t}=\frac{\left(\sum_{i=1}^Np_i^t\right)^N}{N^N\prod_{i=1}^Np_i^t},\quad\forall t\in\N.
\end{equation*}
Moreover, the AM-GM inequality exactly claims that
\[
\frac{\left(\sum_{i=1}^Np_i\right)^N}{N^N\prod_{i=1}^Np_i}\geq 1,\quad\forall \mathbf p\in (0,1)^N,
\]
with equality iff all $\mathbf p$-coordinates are equal, viz.\ iff $\frac{p_i}{p_1}=1$ for all $i\in \{2,\cdots ,N\}$. Therefore, in every orbit, the sequence $\{\pi^t\}_{t\in\N}$ is non-decreasing. Actually, together with Claim \ref{NONSYNC}, this argument shows that the sequence cannot remain constant in non-synchronized orbits. Moreover, assume that there exists $\epsilon>0$ such that 
\begin{equation}
\min_{i\in \{2,\cdots ,N\}}|\frac{p^t_i}{p^t_1}-1|\geq \epsilon,
\label{LOBOUNDP}
\end{equation}
for an arbitrary large number of instances of $t$. At these instances, we have $\frac{\pi^{t+1}}{\pi^t}\geq 1+\epsilon'$ for some $\epsilon'>0$, implying altogether that the sequence $\{\pi^t\}_{t\in\N}$ must grow with exponential rate. In particular, one can make sure that $\pi^t>1$, and therefore $\max_ia_i^t>1$ as desired, for $t$ sufficiently large. 

We are going to prove that the lower bound \eqref{LOBOUNDP} holds for an arbitrary large number of instances of $t$, in the system for which the $\mathbf p$-coordinates are iterated using the linearized dynamics at $\mathbf p=\mathbf 0$ for $\mathbf a\in (0,1)^N$ and $\alpha=0$. Then we shall use that under the assumption (otherwise there is nothing to prove)
\[
\max_{i\in\{1,\cdots ,N\}}a_i^t<1,\ \forall t\in\N,
\]
an arbitrary large number of iterates under \eqref{DEFDYNAM} must remain close to the corresponding ones under the linearization, provided that $\mathbf p^0$ lies in a sufficiently small neighbourhood of $\mathbf 0$ (which depends on $\mathbf a^0$ and $\alpha$).   

Given that $f'_a(0)=a$ for $a<1$, for $g(p,q)=\frac{q}{p}$, the system obtained by linearizing the first equation in \eqref{DEFDYNAM} at $\mathbf p^t=\mathbf 0$ and for $\mathbf a^{t+1}\in (0,1)^N$ and $\alpha=0$ is given by 
\begin{equation}
\left\{\begin{array}{l}
p_i^{t+1}=a_i^{t+1}p_i^t\\
a_i^{t+1}=a_i^t\frac{\sum_{i=1}^Np_i^t}{Np_i^t}
\end{array}\right.,\ \forall i\in\{1,\cdots ,N\}.
\label{LINEARIZED}
\end{equation}
Letting $\rho_i^t=\frac{p_i^t}{p_1^t}$ and $\gamma_i^t=\frac{a_i^t}{a_1^t}$, \eqref{LINEARIZED} implies the following iteration rule for the variables $(\rho_2,\cdots ,\rho_N,\gamma_2,\cdots ,\gamma_N)$
\[
\left\{\begin{array}{l}
\rho_i^{t+1}=\gamma_i^{t}\\
\gamma_i^{t+1}=\frac{\gamma_i^t}{\rho_i^t}
\end{array}\right.\quad \forall i\in\{2,\cdots ,N\}
\]
Accordingly, we must have 
\[
\rho_i^{t+3}=\frac1{\rho_i^t}\quad\text{and}\quad \gamma_i^{t+3}=\frac1{\gamma_i^t}\quad \forall t\in\N,
\]
hence, the sequence $\{\rho_i^t\}_{t\in\N}$ must be periodic with period (at most) 6. Therefore, if the coordinates of $\mathbf p^0$ are not all equal, there must exist $i\in\{1,\cdots ,N\}$ such that 
\[
\rho_i^{6k}=\rho_i^0\neq 1,\quad\forall k\in\N,
\]
proving that \eqref{LOBOUNDP} holds for all $t\in 6\N$ as desired. 

Consequently, for every $\mathbf a^0\in (0,1)^N$ and every $\mathbf p^0\in (0,1)^N$ whose coordinates are not all equal, there exist $t(\mathbf p^0,\mathbf a^0)\in\N$ such that 
\[
 \max_{t\in\{0,\cdots ,t(\mathbf p^0,\mathbf a^0)-1\}}\max_{i\in\{1,\cdots ,N\}}a_i^{t}\leq 1<\max_{i\in\{1,\cdots ,N\}}a_i^{t(\mathbf p^0,\mathbf a^0)},
\]
for the orbit generated by \eqref{LINEARIZED}.  Notice that $t(\epsilon\mathbf p,\mathbf a)$ does not depend on $\epsilon\in \R_\ast^+$ because \eqref{LINEARIZED} commutes with the homogeneous scaling $\epsilon\mapsto \epsilon\mathbf p$ of the $\mathbf p$-coordinates.

Let $F_\text{lin}$ be the multidimensional map associated with \eqref{LINEARIZED}, {\sl ie.} $(\mathbf p^{t+1},\mathbf a^{t+1})=F_\text{lin}(\mathbf p^t,\mathbf a^t)$. Given $\mathbf x=(x_1,\cdots ,x_{N})\in\R^{N}$ let 
\[
\|\mathbf x\|_N:=\max_{i\in\{1,\cdots ,N\}}|x_i|.
\]
\begin{Claim}
Given $\mathbf a\in (0,1)^N$, $t\in\N$ and $\epsilon>0$, there exists $\delta>0$ such that for every $\mathbf p\in (0,1)^N$ satisfying $\|\mathbf p\|_N<\delta$, we have 
\[
\max_{k\in \{1,\cdots ,t\}}\|F^k(\mathbf p,\mathbf a)-F_\text{lin}^k(\mathbf p,\mathbf a)\|_{2N}<\epsilon,
\]
provided that all $\mathbf a$-coordinates of each iterate $F_\text{lin}^k(\mathbf p,\mathbf a)$ lie in $(0,1)$.
\label{CONTORBIT}
\end{Claim}
This statement, whose proof will be given below, helps to complete the proof of the Proposition. Let $\mathbf a^0\in (0,1)^N$ and $\mathbf p\in (0,1)^N$ be such that $(\mathbf p,\mathbf a^0)$ is not synchronized. Let $t_0\in\N$ be the first time such that the coordinates of $\mathbf p^{t_0}$ are not all equal. We may assume that $\mathbf a^{t_0}\in (0,1)^N$, otherwise there is nothing to prove. For the sake of notation, we also assume $t_0=0$. Now, let 
\[
\epsilon_0=\max_{i\in\{1,\cdots ,N\}}a_{\text{lin},i}^{t(\mathbf p,\mathbf a^0)}-1,
\]
where $\mathbf a_\text{lin}^t$ are the $\mathbf a$-coordinates of $F_\text{lin}^t(\mathbf p, \mathbf a^0)$. Let also $\Delta>0$ be the quantifier $\delta$ given by Claim \ref{CONTORBIT} when $t=t(\mathbf p,\mathbf a^0)$ and $\epsilon=\epsilon_0$. For every $\delta\in (0,\Delta)$, the orbit of $F$ issued from $(\delta \mathbf p,\mathbf a^0)$ certainly satisfies
\[
\max_{i\in\{1,\cdots ,N\}}a_i^{t(\mathbf p,\mathbf a^0)}>1,
\]
as desired, where $\mathbf a^t$ are the $\mathbf a$-coordinates of $F^t(\delta \mathbf p,\mathbf a^0)$. The proof of Proposition \ref{INSTAB} is complete. \hfill $\Box$
\medskip

\noindent
{\sl Proof of Claim \ref{CONTORBIT}.} The proof is standard and proceeds by induction based on the following properties.
\begin{itemize}
\item[(i)] For every $(\mathbf p,\mathbf a)\in (0,1)^N\times (\R_\ast^+)^N$ and $\epsilon>0$, there exists $\delta>0$ such that for every $(\mathbf p',\mathbf a')$ satisfying $\|(\mathbf p',\mathbf a')-(\mathbf p,\mathbf a)\|_{2N}<\delta$, we have 
\[
\|F(\mathbf p',\mathbf a')-F(\mathbf p,\mathbf a)\|_{2N}<\epsilon,
\] 
as a direct consequence of the facts that $a\mapsto f_a(x)$ and $a\mapsto \|f'_a\|_\infty$ are continuous in $\R_\ast^+$,\footnote{which implies that 
\[
|f_{a'}(p')-f_a(p)|\leq |f_{a'}(p')-f_{a'}(p)|+|f_{a'}(p)-f_a(p)|,
\]
can be made arbitrarily small by taking $(p',a')$ arbitrarily close to $(p,a)$.} together with continuity of the map $g$ at every point of $(0,1]\times [0,1]$.
\item[(ii)] For every $\mathbf a\in (0,1)^N$ and $\epsilon>0$, there exists $\delta>0$ such that for every $\mathbf p\in (0,1)^N$ satisfying $\|\mathbf p\|_N<\delta$, we have
\[
\|F(\mathbf p,\mathbf a)-F_\text{lin}(\mathbf p,\mathbf a)\|_{2N}<\epsilon,
\] 
as an immediate consequence of the fact that $f_a\in C^2([0,1])$ for every $a\in \R_\ast^+$.
\end{itemize}
For the induction, consider the following decomposition
\[
F^{t+1}(\mathbf p,\mathbf a)-F_\text{lin}^{t+1}(\mathbf p,\mathbf a)=F^{t+1}(\mathbf p,\mathbf a)-F\circ F_\text{lin}^{t}(\mathbf p,\mathbf a)+F\circ F_\text{lin}^{t}(\mathbf p,\mathbf a)-F_\text{lin}^{t+1}(\mathbf p,\mathbf a),
\]
and use 
\begin{itemize}
\item the induction hypothesis together with property {\sl (i)} above in order to control the first difference,
\item  and the fact that all $\mathbf p$-coordinates of $F_\text{lin}^{t}(\mathbf p,\mathbf a)$ remain small (also shown by induction) together with {\sl (ii)} above (and $\mathbf a_\text{lin}^t\in (0,1)^N$) in order to control the second difference. \hfill $\Box$
\end{itemize}

\subsection{Global analysis in phase space}\label{P-MAINRES}
This Section is devoted to the proof of Theorem \ref{MAINRES} and begins by the presentation and proof of several auxiliary properties.

\subsubsection{Preliminary statements}
We begin with an easy preliminary result on the behaviour of the iterates
 \begin{equation}
 p_{t+1}=f_{\alpha ,a_{t+1}}(p_t),\ t\in\N,
 \label{AUXILIT}
 \end{equation}
given a sequence $\textbf{a}=\{a_t\}_{t\in\N}\in (\R^+_\ast)^\N$. While the result is elementary in itself, it serves as a reference for the results about the collective system.
\begin{Claim}
Assume that $\limsup_{t\to +\infty}a_t<1$ (resp.\ $\liminf_{t\to +\infty}a_t>1$). Then, under the iterations \eqref{AUXILIT} above, the iterates have the following asymptotic behaviour
\[
\lim_{t\to +\infty} p_t=0\quad (\text{resp.}\ \lim_{t\to +\infty} p_t=1)\quad \text{for all}\ p_0\in (0,1).
\]
\label{CONVFA}
\end{Claim}
\begin{proof}
We only prove the first case. The second one follows from symmetry. 
By assumption, let $T\in\N$ be such that $a_t<1$ for all $t\geq  T$. For $t\geq T-1$, we have $f_{\alpha ,a_{t+1}}(p)<p$ for all $p\in (0,1)$; hence the sequence $\{p_t\}_{t\in\N}$ eventually decreases. Therefore, it must be convergent. 
Let us show that the limit $p_\infty \in [0,1)$ must be equal to 0. Let $\{a_{t_n}\}_{n\in\N}$ be a convergent subsequence and let $a_\infty$ be its limit. Writing
\begin{align*}
|f_{\alpha ,a_{t_n}}(p_{t_n-1})-f_{\alpha ,a_\infty}(p_\infty)|&\leq |f_{\alpha ,a_{t_n}}(p_{t_n-1})-f_{\alpha ,a_{t_n}}(p_\infty)|+|f_{\alpha ,a_{t_n}}(p_\infty)-f_{\alpha ,a_\infty}(p_\infty)|\\
&\leq \|f'_{\alpha ,a_{t_n}}\|_\infty|p_{t_n-1}-p_\infty|+|f_{\alpha ,a_{t_n}}(p_\infty)-f_{\alpha ,a_\infty}(p_\infty)|
\end{align*}
and using the continuity of $a\mapsto \|f'_a\|_\infty$, we obtain 
\[
p_\infty=\lim_{n\to\infty}p_{t_n}=\lim_{n\to\infty}f_{\alpha ,a_{t_n}}(p_{t_n-1})=f_{\alpha ,a_\infty}(p_\infty),
\]
hence $p_\infty=0$. 
\end{proof}

The next statement collects some constraints on the asymptotic behaviours in the system \eqref{DEFDYNAM}.
\begin{Claim} (i) It is impossible that 
\[
\min_{i\in\{1,\cdots ,N\}}\limsup_{t\to +\infty}a_i^t<1<\max_{i\in\{1,\cdots ,N\}}\liminf_{t\to +\infty}a_i^t.
\]

\noindent
(ii) Assume that all the sequences $\{p_i^t\}_{t\in\N}$ converge. Then, their limits must be equal.
\label{EQCONV}
\end{Claim}
\begin{proof}
{\sl (i)} By contradiction, assume the existence of $i_\text{min},i_\text{max}\in\{1,\cdots ,N\}$ such that 
\[
\limsup_{t\to +\infty}a_{i_\text{min}}^t<1<\liminf_{t\to +\infty}a_{i_\text{max}}^t.
\]
Claim \ref{CONVFA} then implies the following limits
\[
\lim_{t\to +\infty}p_{i_\text{min}}^t=0\quad\text{and}\quad \lim_{t\to +\infty}p_{i_\text{max}}^t=1.
\]
As a consequence, the continuity of $g$ and the inequality \eqref{INEQG} imply the existence of $\epsilon>0$ and  $T\in\N$ such that
\[
g(p_{i_\text{min}}^t,\frac1{N}\sum_{i=1}^{N}p_i^t)>1+\epsilon>1>1-\epsilon>g(p_{i_\text{max}}^t,\frac1{N}\sum_{i=1}^{N}p_i^t),\quad \forall t>T,
\]
from which it follows that 
\begin{equation}
\lim_{t\to +\infty}a_{i_\text{min}}^t=+ \infty\quad\text{and}\quad \lim_{t\to +\infty} a_{i_\text{max}}^t=0,
\label{XTRLIM}
\end{equation}
in contradiction with the initial assumption. 
\medskip

\noindent
{\sl (ii)} Similarly, assume the existence of $i_\text{min},i_\text{max}\in\{1,\cdots ,N\}$ such that 
\[
p_{i_\text{min}}^{\infty} := \lim_{t\to +\infty}p_{i_\text{min}}^t<p_\text{max}^{\infty} := \lim_{t\to +\infty}p_{i_\text{max}}^t.
\]
Then, as before, the limits \eqref{XTRLIM} must hold and therefore we must have $p_{i_\text{min}}^\infty=1$ and $p_{i_\text{max}}^\infty=0$, 
contradicting the inequality above.
\end{proof}

\subsubsection{Proof of Theorem \ref{MAINRES}}
The Theorem is an immediate corollary of the following more technical statement.
\begin{Pro}
Assume that $g$ satisfies the conditions \eqref{UPBOUNDG} and \eqref{CONCAV}. Then the following statements hold.

\noindent
(i) In every orbit, the product $\pi^t$ is non-increasing.

\noindent
(ii) If $\lim_{t\to +\infty} \pi^t>0$, then we have 
\[
\lim_{t\to +\infty}\max_{i,j\in\{1,\cdots ,N\}}|p_i^t-p_j^t|=0.
\]

\noindent
(iii) The condition
\[
\liminf_{t\to +\infty}\min_{i\in\{1,\cdots ,N\}}|a_i^t-1|>0,
\]
implies that the corresponding orbit must converge either to a fixed point or to $(\mathbf 0,\mathbf 0)$.

\noindent
(iv) For every orbit, we have
\[
\sup_{t\in\N}\max_{i\in\{1,\cdots ,N\}}a_i^t<+\infty.
\]
\label{LASTP}
\end{Pro}
\begin{proof}
{\em (i)} Applying the AM-GM inequality and then the inequality \eqref{CONCAV} yields
\[
\prod_{i=1}^Ng(p_i,\frac1{N}\sum_{i=1}^Np_i)\leq\left(\frac1{N}\sum_{i=1}^Ng(p_i,\frac1{N}\sum_{i=1}^Np_i)\right)^N\leq 1,\ \forall \mathbf p\in (0,1)^N.
\]
\medskip

\noindent
{\em (ii)} We have 
\[
\lim_{t\to +\infty} \frac{\pi^{t+1}}{\pi^t}=\lim_{t\to +\infty} \prod_{i=1}^Ng(p_i^t,\frac1{N}\sum_{i=1}^Np_i^t)=1,
\]
which implies 
\begin{equation}
\lim_{t\to +\infty} \frac{a_i^{t+1}}{a_i^t}=\lim_{t\to +\infty} g(p_i^t,\frac1{N}\sum_{i=1}^Np_i^t)=1,\ \forall i\in\{1,\cdots ,N\}.
\label{CONVRAT}
\end{equation}
Indeed, by compactness, let $\{t_k\}_{k\in\N}$ be an infinite subsequence such that  
\[
\lim_{k\to +\infty} g(p_i^{t_k},\frac1{N}\sum_{i=1}^Np_i^{t_k})=g_i^\infty\in [0,\gamma]\ \text{exists for all}\ i\in\{1,\cdots ,N\}.
\]
Then, we have
\[
1=
\prod_{i=1}^Ng_i^\infty\leq \left(\frac1{N}\sum_{i=1}^Ng_i^\infty\right)^N=\lim_{k\to +\infty} \left(\frac1{N}\sum_{i=1}^Ng(p_i^{t_k},\frac1{N}\sum_{i=1}^Np_i^{t_k})\right)^N\leq 1,
\]
and optimality of the AM-GM inequality imposes that 
\[
g_i^\infty=1, \ \forall i\in\{1,\cdots ,N\},
\]
as desired. 

Furthermore, we claim that the second limits in \eqref{CONVRAT} implies
\[
\lim_{t\to +\infty} \left|p_i^t-\frac1{N}\sum_{i=1}^Np_i^t\right|=0,\ \forall i\in\{1,\cdots ,N\},
\]
from where the desired conclusion immediately follows. Indeed, by contradiction, assume that $\{\mathbf p^t\}$ had an accumulation point $\mathbf p^\infty$ such that 
\[
\min_{i\in\{1,\cdots ,N\}}\left|p_i^\infty- \frac1{N}\sum_{i=1}^Np_i^\infty\right|>0.
\]
The continuity of the map $g$ implies that each sequence $\{g(p_i^t,\frac1{N}\sum_{i=1}^Np_i^t)\}$ must accumulate on $g(p_i^\infty,\frac1{N}\sum_{i=1}^Np_i^\infty)$. However, one of these values would have to differ from 1 because of the condition \eqref{INEQG}, which is impossible.
\medskip

\noindent
{\em (iii)} We separate the cases $\lim_{t\to +\infty} \pi^t=0$ and $\lim_{t\to +\infty} \pi^t>0$. In the first case, let $\{t_k\}_{k\in\N}$ be an infinite subsequence such that 
\[
\lim_{k\to +\infty}\min_{i\in \{1,\cdots ,N\}}|a_i^{t_k}-1|>0.
\]
By passing to a further subsequence if necessary, the convergence $\pi^t\to 0$ implies the existence of $i_0\in\{1,\cdots ,N\}$ such that 
\[
\lim_{k\to +\infty} a_{i_0}^{t_k}=0.
\]
By continuity and compactness, let\footnote{For the example $g(p,q)=\frac{q}{p}$, we actually have
\[
\sup_{\mathbf p\in (0,1)^N}Sg(p_i,\frac1{N}\sum_{i=1}^Np_i)<N,
\]
which is all we need for our purpose.}
\[
\gamma=\max_{[0,1]^2}Sg<+\infty.
\]
Together with the previous limit, that $\gamma$ is finite implies that for every $t\in\N$, we have 
\[
\lim_{k\to +\infty} \max_{s\in\{0,\cdots ,t\}}a_{i_0}^{t_k-s}=0.
\]
Independently, the properties of the maps $f_a$ ensure that 
\[
\lim_{t\to +\infty}\max_{p\in [0,1]}f^t_{\alpha,a}(p)=0,\ \forall a\in (0,1).
\]
These two limits imply the following one
\[
\lim_{k\to +\infty} p_{i_0}^{t_k}=0.
\]
Besides, we may assume that $\{t_k\}_{k\in\N}$ is such that $a_{i_0}^{t_k-1}\geq a_{i_0}^{t_k}$. Indeed, it is impossible that $\{a_{i_0}^t\}_{t\in\N}$ be eventually increasing. Moreover, its liminf is certainly attained in the set of values of $t$ for which $a_{i_0}^{t-1}\geq a_{i_0}^t$. The inequality $a_{i_0}^{t_k-1}\geq a_{i_0}^{t_k}$ forces the following one
\[
p_{i_0}^{t_k}\geq \frac1{N}\sum_{i=1}^Np_i^{t_k}.
\]
Hence, we must have
\[
\lim_{k\to +\infty} p_i^{t_k}=0,\ \forall i\in\{1,\cdots ,N\}.
\]
Moreover, as shown next, this imposes in turn
\[
\limsup_{k\to +\infty} a_i^{t_k}\leq 1,\ \forall i\in\{1,\cdots ,N\},
\]
and hence 
\[
\lim_{k\to +\infty} p_i^{t_k}=0\quad\text{and}\quad \limsup_{k\to +\infty} a_i^{t_k}< 1,\ \forall i\in\{1,\cdots ,N\},
\]
given the initial assumption in {\em (iii)}. Therefore $(\mathbf p^{t_k},\mathbf a^{t_k})$ must satisfy the conditions of Proposition \ref{STABP1} for $k$ sufficiently large, , hence the orbit $(\mathbf p^{t},\mathbf a^{t})$ must converge, completing the proof in the case $\lim_{t\to +\infty} \pi^t=0$. We prove the claimed inequality on the limsup above by contradiction. Given $i\in\{1,\cdots , N\}$, assume the existence of $\epsilon>0$ such that for every $k\in\N$, there exists $k'>k$ such that $a_i^{t_{k'}}>1+\epsilon$. Then 
\[
p_i^{t_{k'}}\geq f_{\alpha,1+\epsilon}(0)>0,
\]
making it impossible that $\lim_{k\to +\infty} p_i^{t_k}=0$. 

In the case $\lim_{t\to +\infty} \pi^t>0$, when combined with the assumption 
\[
\liminf_{t\to +\infty}\min_{i\in\{1,\cdots ,N\}}|a_i^t-1|>0,
\]
and Claim \ref{EQCONV} - {\sl (i)}, the fact that 
\[
\lim_{t\to +\infty} \frac{a_i^{t+1}}{a_i^t}=1,\ \forall i\in\{1,\cdots ,N\}
\]
implies the existence of $T\in\N$ and $\epsilon>0$ such that we have
\[
\text{either}\ \sup_{t>T}\max_{i\in\{1,\cdots ,N\}}a_i^t\leq 1-\epsilon,\quad \text{or}\ \inf_{t>T}\min_{i\in\{1,\cdots ,N\}}a_i^t\geq 1-\epsilon.
\]
Claim \ref{CONVFA} then ensures convergence to a fixed point. 
\medskip

\noindent
{\em (iv)} By contradiction, assume the existence of $i_1\in\{1,\cdots ,N\}$ and an infinite subsequence $\{t_k\}_{k\in\N}$ such that 
\[
\lim_{k\to +\infty} a_{i_1}^{t_k}=+\infty.
\]
Then similar arguments as in the proof of {\sl (iii)} can be developed to show that we must have
\[
\lim_{k\to +\infty} p_{i_1}^{t_k}=1.
\]
However, since $\pi^t$ converges to a finite limit, there must exist $i_0\in\{1,\cdots ,N\}$ such that
\[
\lim_{k\to +\infty} a_{i_0}^{t_k}=0,
\]
for the same subsequence. Moreover, we showed in the proof of {\sl (iii)} that this limit implies 
\[
\lim_{k\to +\infty} p_i^{t_k}=0,\ \forall i\in\{1,\cdots ,N\},
\]
which contradicts the limit of $\{p_{i_1}^{t_k}\}$ above. 
\end{proof} 

\section{Conclusion}\label{S-DISCUSS}
In this paper, a population model for the dynamics of buyers in markets of perishable goods has been introduced and mathematically analyzed. In particular, a simple mechanism for merchants feedback has been included, based on assuming immediate profit/competitiveness optimisation and prompt estimation of the overall volumes of buyers. 
Such negative regulation implies oscillatory behaviours and makes (local) asymptotic stabilization dependent on the nature of the reactivity rate. A bounded rate cannot counteract intrinsic response to merchants attractiveness/repulsiveness. In such case of moderate reactivity, convergence towards stationary functioning modes takes place, with either full or absence of clientele at all merchants. On the other hand, an arbitrarily large reactivity can promote resilience against vanishing clientele, as shown with the example $g(p,q)=\frac{q}{p}$. 

In addition to the control of local asymptotic stability of the fixed points, a concavity-type condition on the reactivity rate has been identified, that forces the oscillations to eventually cease in every orbit for which the merchants attractiveness remain bounded away from 1, a feature we believe to be generic and to fail only in exceptional cases. In particular, this may fail and oscillations may perdure forever when on the boundary of the fixed points basins of attraction, when neither convergence to a fixed point with $\mathbf p=\mathbf 1$ or $\mathbf p=\mathbf 0$ holds. Moreover, as evidenced in the numerics, to predict the asymptotic functioning mode, either full or empty clientele, is a delicate task when starting close to such boundary. In other words, while the basic ingredients of the dynamics are rather simple, the resulting temporal process that they generate can be rather involved and hardly predictable, especially in case of a major perturbation such as the introduction of a new seller in the market.

We conclude the paper with few additional comments about the modelling assumptions in the system \eqref{DEFDYNAM} and some suggestions about possible improvements. 

First, the modelling time scale has been chosen to be rather coarse for simplicity. It does not incorporate intraday variations of clientele or attractiveness. However, evidences have been given, especially at Rungis market \cite{C10}, that sellers may react during sessions in order to adapt their attractiveness to initially scarce clienteles or to insufficient sale volumes given their stocks. Hence, a more elaborated modelling should incorporate more rapid changes such as hourly variations. 

Likewise, the feedback term in \eqref{DEFDYNAM} assumes instant evaluation by the merchants of the overall clientele volume in the market. While this hypothesis is justified when the market place is compact or the number of competitors is limited, in large markets with a high number of merchants selling goods of the same nature, it would be relevant to adopt, in the feedback evaluation, a local clientele estimation that is evaluated in a limited neighbourhood of the seller under consideration. In such case of local interaction terms, the dynamics is often richer than in globally coupled models, as for instance, spatial patterns and non-synchronous solutions emerge \cite{CF05}.

Equally important since real populations are of finite size, it is questionable to consider  fractions $p_i^t$ that are extremely small or extremely large, and in particular smaller than the inverse of the total population size when the latter is given. A description of the dynamics at the level of individuals should be adopted when $p_i^t$ is close to 0 or to 1, especially when experimenting the introduction of a new seller with very small initial clientele. In that setting, it would be particularly interesting to identify the analogues of the stability condition \eqref{UPBOUNDG} at the level of individual dynamics. 

In addition, a more detailed dynamics at the level of individuals gives the opportunity to take into account intrinsic heterogeneities in the individual features that may arise in such diversified populations of buyers that are present in these markets. Likewise, temporal non-systematic fluctuations of the individual behaviours could be integrated.  In particular both these aspects could be taken into account through modelling of the dynamics using random process. All these considerations can be the subject of future studies.  

\paragraph{Acknowledgments.} 
We thank Simon Bussy, Corentin Curchod, Adeline Fermanian and Annick Vignes for fruitful discussions and relevant comments. We also acknowledge the anonymous referee for constructive critiques and suggestions. The research in this article has been accomplished in the framework of the project ANR LabCom LOPF, ref ANR-20-LCV1-0005.

\end{document}